\documentclass[twoside]{article}
\usepackage[a4paper]{geometry}
\usepackage[utf8,latin1]{inputenc} 
\usepackage[OT1,T1]{fontenc} 
\usepackage{RR}
\usepackage{hyperref}
\usepackage{amssymb}
\usepackage{amsfonts}
\usepackage{amsmath}
\usepackage{epstopdf}
\usepackage{graphicx}

\newcommand {\R} {\rm I\!R}

\newtheorem{condition}{\textbf{Condition}}{\bfseries}{\itshape}
\newtheorem{definition}{Definition}

\newtheorem{proof}{Proof}
\newtheorem{theorem}{Theorem}
\newtheorem{remark}{Remark}
\usepackage[frenchb]{babel} 
\RRNo{8403}
\RRdate{November 2013}
\RRauthor{
Konstantin Avrachenkov\thanks{INRIA Sophia Antipolis,
2004 Route des Lucioles, Sophia Antipolis, France, +33(4)92387751, k.avrachenkov@sophia.inria.fr}%
  \and
Oussama Habachi\thanks{University of Avignon, 339 Chemin des Meinajaries,
Avignon, France, +33(4)90843518, oussama.habachi@univ-avignon.fr}%
\and Alexey Piunovskiy\thanks{Department of Mathematical Sciences, University of
Liverpool, Liverpool, UK, +44(151)7944737, piunov@liv.ac.uk}
\and Yi Zhang\thanks{Department of Mathematical Sciences, University of
Liverpool, Liverpool, UK, +44(151)7944761, zy1985@liv.ac.uk}

}
\RRtitle{Th\'eorie du contr\^ole optimal impulsif \`a horizon infini avec application au contr\^ole de congestion dans Internet}
\RRetitle{Infinite Horizon Optimal Impulsive Control Theory with Application to Internet Congestion Control}
\titlehead{Infinite Horizon Optimal Impulsive Control Theory with Application to Internet Congestion Control}
\RRresume{
Dans ce papier, nous d\'eveloppons une approche bas\'ee sur les \'equations de Bellman  pour les probl\`emes de contr\^ole optimal impulsif  \`a horizon infini. Les deux crit\`eres actualis\'e et moyenne sur le temps sont pris en consid\'eration.
Nous \'etablissons des conditions naturelles et tr\`es g\'en\'erales  en vertu desquelles un triplet canonique de contr\^ole
produit une politique de r\'etroaction optimale.
Ensuite, nous appliquons nos r\'esultats g\'en\'eraux
 pour le contr\^ole de congestion dans  Internet offrant un cadre id\'eal
pour la conception des algorithmes de gestion active de files d'attente.
En particulier,
nos r\'esultats th\'eoriques g\'en\'eraux sugg\`erent un syst\`eme de gestion active de files d'attente  \`a seuil
 qui prend en compte les principaux param\`etres
du protocole de contr\^ole de transmission.
}
\RRabstract{
We develop Bellman equation based approach for infinite time
horizon optimal impulsive control problems. Both discounted
and time average criteria are considered. We establish very general and at the same
time natural conditions under which a canonical control triplet
produces an optimal feedback policy. Then, we apply our general
results for Internet congestion control providing a convenient setting
for the design of active queue management algorithms. In particular,
our general theoretical results suggest a simple threshold-based
active queue management scheme which takes into account the main parameters
of the transmission control protocol.
 }
\RRmotcle{Contr\^ole optimale impulsif; Horizon de temps infini;
Crit\`eres actualis\'e et moyenne sur le temps; Alpha-\'equit\'e; Contr\^ole de congestion dans Internet.}
\RRkeyword{Optimal impulsive control; Infinite time horizon;
Average and discounted criteria; Alpha-fairness; Internet Congestion Control. }
\RRprojets{MAESTRO}
\RCSophia 

\begin{document}
\makeRR   
\section{Introduction}

Recently, there has been a steady increase in the demand for QoS
(Quality of Services) and fairness among the increasing number of
IP (Internet Protocol) flows. With respect to QoS, a plethora of
research focuses on smoothing the throughput of AIMD (Additive
Increase Multiplicative Decrease)-based congestion control for the
Transmission Control Protocol (TCP), which is prevailingly
employed in today's transport layer communication. These
approaches adopt various congestion window updating policies to
determine how to adapt the congestion window size to the network
environment. Besides, there have been proposals of new high speed
congestion control algorithms that can efficiently utilize the
available bandwidth for large volume data transfers, see
\cite{Floyd03,TKelly04,Leith04,Wei06}. Although TCP gives
efficient solutions to end-to-end error control and congestion
control, the problem of fairness among flows is far from being
solved. See for example, \cite{AAP05,Moller07,Li07} for the
discussions of the unfairness among various TCP versions.

The fairness can be improved by the Active Queue Management (AQM)
through the participation of links or routers in the congestion
control. The first AQM scheme, the Random Early Drop (RED), is
introduced in \cite{RED}, and allows to drop packets before the
buffer overflows. The RED was followed by a plethora of AQM
schemes; a survey of the most recent AQM schemes can be found in
\cite{Adams}. However, the improvement in fairness provided by
AQMs is, on the one hand, still not satisfactory; and, on the
other hand, at the core of the present paper.

Since most of the currently operating TCP versions exhibit a saw-tooth
like behavior, it appears that the setting of impulsive control
is very well suited for the Internet congestion control.
Furthermore, since the end users expect
permanent availability of the Internet, it looks natural to consider
the infinite time horizon setting. With the best of our efforts, we
could not find any available results about infinite time horizon
optimal impulsive control problems. Thus, a general theory for infinite
time horizon optimal impulsive control needs to be developed.
We note that the results available in \cite{Miller03} and references
therein about finite time horizon optimal impulsive control problems
cannot directly be applied to the infinite horizon with non-decreasing
energy of impulses. In \cite{Miller03} the impulsive control is described
with the help of Stieltjes integral with respect to bounded variation
function. Clearly, the bounded variation function cannot represent
an infinite number of impulses with non-decreasing energy.

Therefore, in the first part of the paper, we develop Bellman equation based
approach for infinite time horizon optimal impulsive control problems.
We consider both discounted and time average criteria. We establish very general
and at the same time natural conditions under which a canonical control triplet
produces an optimal feedback policy.

Then, in the second part of the paper we apply the developed general results
to the Internet congestion control. The network performance is
measured by the long-run average $\alpha$-fairness and the
discounted $\alpha$-fairness, see \cite{Mo00}, which can be specified to
the total throughput, the proportional fairness and the max-min
fairness maximization with the particular values of the tuning
parameter $\alpha$. The model in the present paper is different
from the existing literature on the network utility maximization
see e.g., \cite{Kunniyur03,Kelly98,Low99}, in at least two
important aspects: (a) we take into account the fine, saw-tooth like,
dynamics of congestion control algorithms, and we suggest the use
of per-flow control and describe its form. Indeed, not long ago a
per-flow congestion control was considered infeasible. However,
with the introduction of modern very high speed routers, the
per-flow control becomes realistic, see \cite{Noirie09}. (b) By solving
rigorously the impulsive control problems, we propose a novel AQM
scheme that takes into account not only the traffic transiting
through bottleneck links but also end-to-end congestion control
algorithms implemented at the edges of the network. More
specifically, our scheme asserts that a congestion
notification (packet drop or explicit congestion notification)
should be sent out whenever the current sending rate is over a threshold,
whose closed-form expression is computed.

The remainder of this paper proceeds as follows. In the next
section we give preliminary results regarding general average and
discounted impulsive optimal control problems. In Section
\ref{TCP}, we describe the mathematical models for the congestion
control, and solve the underlying optimal impulsive control problems
based on the results obtained in Section \ref{TCP}. Section \ref{con}
concludes the paper.

\section{Preliminary result}\label{opt}
In this section, we establish the verification theorems for a
general infinite horizon impulsive control problem under the
long-run average criterion and the discounted criterion, which are
then used to solve the concerned Internet congestion control problems
in the next section.

\subsection{Description of the controlled process}
Let us consider the following dynamical system in $X\subseteq
\R^n$ (with $X$ being a nonempty measurable subset of
$\mathbb{R}^n$, and some initial condition $x(0)=x_0\in X$)
governed by
\begin{equation}\label{e1}
dx=f(x,u)dt,
\end{equation}
where $u\in U$ is the gradual control, with $U$ being an arbitrary
nonempty Borel space. Suppose another nonempty Borel space $V$ is
given, and, at any time moment $T$, if he decides so,
 {the decision maker} can apply an impulsive
control $v\in V$ leading to the following {new state}:
\begin{equation}\label{e2}
x(T)=j(x(T^-),v),
\end{equation}
where $j$ is a measurable mapping from $\mathbb{R}^n\times V$ to
$X.$ Thus, we have the next definition of a policy.

\begin{definition}\label{r1}
A policy $\pi$  is defined by a $U$-valued measurable {mapping}
$u(t)$ and a sequence of impulses $\{T_i,v_i\}_{i=1}^\infty$ with
$v_i\in V$ and $\dots\ge T_{i+1}\ge T_i\ge 0,$ which satisfies
$T_0:= 0$ and $\lim_{i\to\infty} T_i=\infty$. A policy $\pi$ is
called a feedback one if one can write{\footnote{Here the
superscript $f$ stands for ``feedback''.}} $u(t)=u^f(x(t))$,
$T_i^{\cal L}=\inf\{t> T_{i-1}:~x(t)\in{\cal L}\}$,
$v_i=v^{f,{\cal L}}(x(T_i^-))$, where $u^{f}$ is a $U$-valued
measurable {mapping} on $\mathbb{R}^n,$ and ${\cal L}\subset X$ is
a specified (measurable) subset of $X$. A feedback policy is
completely characterized and thus denoted by the triplet
$(u^f,{\cal L}, v^{f,{\cal L}})$.
\end{definition}

We are interested in the (admissible) policies $\pi$ under which
the following hold (with any initial state). (a) $T_0\le
T_1<T_2<\ldots$. \footnote{In the case when two (or more) impulses
$v_i$ and $v_{i+1}$ are applied simultaneously, that is
$T_{i+1}=T_i$, we formulate this as a single impulse $\hat v$ with
the effect $j(x,\hat v):=j(j(x,v_i),v_{i+1})$, and include $\hat
v$ into the set $V$.} (b) The controlled process $x(t)$ described
by (\ref{e1}) and (\ref{e2}) is well defined: for any initial
state $x(0)=x_0$, there is a unique piecewise differentiable
function $x^\pi(t)$ with $x^\pi(0)=x_0$, satisfying (\ref{e1}) for
all $t,$ wherever the derivative exists; satisfying (\ref{e2}) for
all $T=T_i$, $i=1,2,\ldots;$ and satisfying that $x^\pi(t)$ is
continuous at each $t\ne T_i.$ (c) Within a finite interval, there
are no more than finitely many impulsive controls. The controlled
process under such a policy $\pi$ is denoted by $x^\pi(t)$.

\subsection{Optimal impulsive control problem and Bellman equation}

Let $c(x,u)$ be the reward rate if the controlled process is at
the state $x$ and the gradual control $u$ is applied, and $C(x,v)$
be the reward earned from applying the impulsive control $v$.
Under the policy $\pi$ and initial state $x_0$, the average reward
is defined by
\begin{eqnarray}\label{Ze1}
J(x_0,\pi)&=&\liminf_{T\rightarrow\infty}\frac{1}{T}\left\{\int_0^T
c(x^\pi(t),u(t))dt+\sum_{i=1}^{N(T)} C(x^\pi(T_i^-),v_i)\right\},
\end{eqnarray}
where and below $N(T):=\sup{\{n>0,T_n\leq T\}}$, and
$x(T_0^-):=x_0$; and the discounted reward (with the discount
factor $\rho>0$) is given by
\begin{equation*}
J_\rho(x_0,\pi)=\liminf_{T\to\infty} J^T_\rho(x_0,\pi),
\end{equation*}
where
\begin{equation}\label{e4}
J^T_\rho(x_0,\pi)=\int_0^T e^{-\rho t}
c(x^\pi(t),u(t))dt+\sum_{i=1,2,\ldots T_i\in[0,T]} e^{-\rho T_i}
C(x^\pi(T_i-0),v_i).
\end{equation}
We only consider the class of (admissible) policies $\pi$ such
that the right side of (\ref{Ze1}) (resp., (\ref{e4})) is well
defined under the average (resp., discounted) criterion, i.e., all
the limits and integrals are finite, which is automatically the
case, e.g., when $C$ and $c$ are bounded functions. The optimal
control problem under the average criterion reads
\begin{eqnarray}\label{ZYPiuK}
J(x_0,\pi)\rightarrow\max_{\pi},
\end{eqnarray}
and the one under the discounted criterion reads
\begin{equation}\label{e3}
J_\rho(x_0,\pi)\to\max_\pi.
\end{equation}
A policy $\pi^\ast$ is called (average) optimal (resp.,
(discounted) optimal) if $J (x_0,\pi^\ast)=\sup_{\pi}J(x_0,\pi )$
(resp., $J_\rho (x_0,\pi^\ast)=\sup_{\pi}J_\rho(x_0,\pi )$) for
each $x_0\in X.$ Below we consider both problems (\ref{Ze1}) and
(\ref{e3}), and provide the corresponding verification theorems
for an optimal feedback policy, see Theorems \ref{ZYt2} and
\ref{t1}.

For the average problem (\ref{Ze1}), we consider the following
condition.
\begin{condition}\label{ZYCA}
There are a continuous function $h(x)$ on $X$ and a constant
$g\in\mathbb{R}$ such that the following hold.

\par\noindent(i) The gradient $\frac{\partial h}{\partial x}$ exists
everywhere apart from a subset ${\cal D}\subset X,$ whereas under
every policy $\pi$ and for each initial state $x_0,$ $h(x^\pi(t))$
is absolutely continuous on $[T_i,T_{i+1})$, $i=0,1,\dots;$ and
$\{t\in[0,\infty): x^\pi(t)\in {\cal D}\}$ is a null set with
respect to the Lebesgue measure.

\par\noindent (ii) For all $x\in X\setminus{\cal D}$,
\begin{eqnarray}\label{ZYBellmanAverage}
&&\max\left\{\sup_{u\in U}\left[ c(x,u)-g+\langle\frac{\partial
h}{\partial x},f(x,u)\rangle \right],~\sup_{v\in V}\left[
C(x,v)+h(j(x,v))-h(x)\right]\right\}=0,
\end{eqnarray}
and for all $x\in {\cal D}$, $ \sup_{v\in V}\left[
C(x,v)+h(j(x,v))-h(x)\right]\le 0. $

\par\noindent (iii) There are a measurable subset ${\cal L}^*\subset X$ and a feedback policy $\pi^*=(u^{f*},{\cal L}^*,
v^{f,{\cal L}^*})$ such that for all $x\in X\setminus({\cal
D}\cup{\cal L}^*),$
  $
  c(x,u^{f*}(x))-g+\langle\frac{\partial h}{\partial x}, f(x,u^{f^*}(x)\rangle=0
  $
and for all $x\in{\cal L^*},$ $
  C(x,v^{f,{\cal L}^*})+h(j(x,v^{f,{\cal L}^*(x)}))-h(x)=0,
$ and $j(x,v^{f,{\cal L}^*}(x))\notin{\cal L^\ast}$.

\par\noindent(iv) For any policy $\pi$ and each
initial state $x_0\in X$, {\small{$\limsup_{T\rightarrow\infty}
\frac{h(x^\pi(T))}{T}\ge 0, $ whereas
$\limsup_{T\rightarrow\infty} \frac{h(x^{\pi^*}(T))}{T}=0.$}}
\end{condition}
Equation (\ref{ZYBellmanAverage}) is the Bellman equation for
problem (\ref{Ze1}). $(g,\pi^\ast,h)$ from Condition \ref{ZYCA} is
called a canonical triplet, and the policy $\pi^\ast$ is called a
canonical policy. The next result asserts that any canonical
policy is optimal for problem (\ref{ZYPiuK}).
\begin{theorem}\label{ZYt2} For the average problem (\ref{ZYPiuK}), any feedback policy $\pi^*$
satisfying Condition \ref{ZYCA} is optimal, and $g$ in Condition
\ref{ZYCA} is the value function, i.e., $g=\sup_{\pi}J(x_0,\pi)$
for each $x_0\in X.$
\end{theorem}
\begin{proof}   For each arbitrarily fixed $T>0$,
initial state $x_0\in X$ and policy $\pi$, it holds that
\begin{eqnarray}\label{ZAverage1}
h(x^\pi(T))&=&h(x_0)+\int_0^T \left\{ \langle\frac{\partial
h}{\partial x}(x^\pi(t)), f(x^\pi(t),u(t))\rangle \right\} dt\nonumber\\
&& +\sum_{i:T_i\in[0,T]} \left\{
h(j(x^\pi(T_i^-),v_i))-h(x^\pi(T_i^-))\right\}.
\end{eqnarray}
Therefore,
  \begin{eqnarray*}
&&\int_0^T c(x^\pi(t),u(t))dt+\sum_{i=1}^{N(T)}
C(x^\pi(T_i^-),v_i)+
h(x^\pi(T))\\
&=&h(x_0)+\int_0^T \left\{ c(x^\pi(t),u(t))+\langle\frac{\partial
h}{\partial x}(x^\pi(t)),
f(x^\pi(t),u(t))\rangle \right\} dt\\
&&+\sum_{i:T_i\in[0,T]} \left\{
h(j(x^\pi(T_i^-),v_i))-h(x^\pi(T_i^-))
+C(x^\pi(T_i^-),v_i)\right\}\le h(x_0)+\int_0^T g dt,
\end{eqnarray*}
where the last inequality is because of (\ref{ZYBellmanAverage})
and the definition of $g$ and $h$ as in Condition \ref{ZYCA}. It
follows that $\frac{1}{T}\left\{\int_0^T
c(x^\pi(t),u(t))dt+\sum_{i=1}^{N(T)} C(x^\pi(T_i^-),v_i)\right\}+
\frac{h(x^\pi(T))}{T}\le \frac{h(x_0)}{T}+g,$ and consequently, $
J(x_0,\pi)+\limsup_{T\rightarrow\infty}\frac{h(x^\pi(T))}{T} \le
g. $ Since {$\limsup_{T\rightarrow\infty}\frac{h(x^\pi(T))}{T}
\ge0$} for each $\pi,$ we obtain $ J(x_0,\pi)\le g $ for each
policy $\pi.$ For the feedback policy $\pi^\ast$ from Condition
\ref{ZYCA}, since
{$\limsup_{T\rightarrow\infty}\frac{h(x^{\pi^\ast}(T))}{T} =0,$}
and we have  $J(x_0,\pi^\ast)= g.$ The statement is proved.
\end{proof}

For the discounted problem (\ref{e3}), we formulate the following
condition.

\begin{condition}\label{c1} There is a continuous function $W(x)$ on $X$
such that the following hold.

\par\noindent (i) Gradient $\frac{\partial W}{\partial x}$ exists everywhere
apart from a subset ${\cal D}\subset X\subset \R^n$; for any
policy $\pi$ and for any initial state $x_0,$ the function
$W(x^\pi(t))$ is absolutely continuous on all intervals
$[T_{i-1},T_i)$, $i=1,2,\ldots$; and the Lebesgue measure of the
set $\{t\in[0,\infty):~x^\pi(t)\in{\cal D}\}$ equals zero.

\par\noindent (ii) The following Bellman equation
\begin{equation}\label{e5}
\max\left\{\sup_{u\in U}\left[ c(x,u)-\rho
W(x)+\langle\frac{\partial W}{\partial x},f(x,u)\rangle
\right],~~\sup_{v\in V}\left[
C(x,v)+W(j(x,v))-W(x)\right]\right\}=0.
\end{equation}
is satisfied for all $x\in X\setminus{\cal D}$ and $\sup_{v\in
V}\left[ C(x,v)+W(j(x,v))-W(x)\right]\le 0$ for all $x\in{\cal
D}.$

\par\noindent (iii) There are a measurable subset ${\cal L}^*\subset X$ and a
feedback policy $\pi^*=(u^{f*}(x),{\cal L}^*, v^{f,{\cal L}^*})$
such that $c(x,u^{f*}(x))-\rho W(x)+\langle\frac{\partial
W}{\partial x}, f(x,u^{f^*}(x)\rangle=0$ for all $x\in
X\setminus({\cal D}\cup{\cal L}^*)$ and
  $C(x,v^{f,{\cal L}^*}(x))+W(j(x,v^{f,{\cal L}^*}(x)))-W(x)=0$
for all $x\in{\cal L^*}$; moreover, $g(x,v^{f,{\cal L}^*}(x))\in
X\setminus{\cal L}^*$.

\par\noindent (iv) For any initial state   $x_0\in X$,
  $\limsup_{T\to\infty} e^{-\rho T} W(x^\pi(T))\ge 0$ for any policy $\pi$
and

$\limsup_{T\to\infty} e^{-\rho T} W(x^{\pi^*}(T))=0.$
\end{condition}

\begin{theorem}\label{t1}
For the discounted problem (\ref{e3}), any feedback policy $\pi^*$
satisfying Condition \ref{c1} is optimal, and $\sup_\pi
J_\rho(x_0,\pi)=W(x_0)=J_\rho(x_0,\pi^*)$ for each $x_0\in X$.
\end{theorem}

\begin{proof}
The proof proceeds along the same line of reasoning as in that of
Theorem \ref{ZYt2}; instead of (\ref{ZAverage1}), one should now
make use of the representation
\begin{eqnarray*}
0&=&W(x_0)+\int_0^T e^{-\rho t}\left\{\langle\frac{\partial
W(x^\pi(t))}{\partial x}, f(x^\pi(t),u(t))\rangle -\rho
W(x^\pi(t))\right\} dt\\
&&+\sum_{i=1,2,\ldots T_i\in[0,T]}e^{-\rho T_i}\left\{
W(g(x^\pi(T_i-0),v_i)-W(x^\pi(T_i-0))\right\}-e^{-\rho
T}W(x^\pi(T)).
\end{eqnarray*}
\end{proof}

\section{Applications of the optimal impulsive control theory to the Internet congestion control}\label{TCP}

In this section, we firstly informally describe the impulsive
control problem for the Internet congestion control, which will then be later formalized
in the framework of the previous section. Let us consider $n$ TCP
connections operating in an Internet Protocol (IP) network of $L$
links defined by a routing matrix $A$, {whose} element $a_{lk}$ is
equal to one if connection $k$ goes through link $l$, {or zero
otherwise.}\footnote{Without loss of generality, we assume that
each link is occupied by some connection, and each connection is
routed through some link.} Denote by $x_k(t)$ the sending rate of
connection $k$ at time $t$. We also denote by $P(k)$ the set of
links corresponding to the path of connection $k$. In this
section, the column vector notation
 {$x(t):=(x_1(t),\dots,x_n(t))^T$} is in use.

The data sources are allowed to use different TCP versions, or if
they use the same TCP, the TCP parameters (round-trip time, the
increase-decrease factors) can be different. Therefore, we suppose
that the sending rate of connection $k$ evolves according to the
following equation
\begin{eqnarray}\label{Kostia1}
\frac{d}{dt} x_k(t) = a_k x^{\gamma_k}(t),
\end{eqnarray}
in the absence of congestion notification, and the TCP reduces the
sending rate abruptly if a congestion notification is sent to the
source $k$, i.e., when a congestion notification is sent to the
source $k$ at time moment $T_{i,k}$ with $T_{0,k}:=0$ and
$T_{i+1,k}\ge T_{i,k},$ its sending rate is reduced as follows
\begin{eqnarray}\label{Kostia2}
x_k(T_{i,k}) = b_k x_k(T_{i,k}^-) < x_k(T_{i,k}^-).
\end{eqnarray}
Here and below, $a_k$, $b_k$ and $\gamma_k$ are constants, which
cover at least two important versions of the TCP end-to-end
congestion control; if $\gamma_k=0$ we retrieve the AIMD
congestion control mechanism (see \cite{A10}), and if $\gamma_k=1$
we retrieve the Multiplicative Increase Multiplicative Decrease
(MIMD) congestion control mechanism (see \cite{TKelly04,Yi10}).
Also note that (\ref{Kostia1}) and (\ref{Kostia2}) correspond to a
hybrid model description that represents well the saw-tooth
behaviour of many TCP variants, see \cite{Hespanha01,A10,Yi10}.

When $T_{i+2,k}>T_{i+1,k}=T_{i,k}> T_{i-1,k}$, multiple (indeed,
two in this case) congestion notifications are being sent out simultaneously at
$T_{i+1,k}=T_{i,k}$; as explained in the previous section, we will
understand such multiple reductions on the sending rate as a
single ``big'' impulsive control. In this section we write
$T_i:=(T_{i,1},\dots,T_{i,n})$ for the $i$th time moments of the
impulsive control for each of the $n$ connections, and assume that
the decision of reducing the sending rate of connection $k$ is
independent upon the other connections. Since there is no gradual
control, we tentatively call the sequence of $T_1,T_2,\dots$ a
policy for the congestion control problem, which will be
formalized below.

We will consider two performance measures of the system; namely
the time average $\alpha$-fairness function
\begin{eqnarray*}
\bar{J}(x_0) = \liminf_{T \to \infty} \frac{1}{1-\alpha}
\sum_{k=1}^n \frac{1}{T} \int_0^T x^{1-\alpha}_k(t) dt,
\end{eqnarray*}
and the discounted $\alpha$-fairness function
\begin{eqnarray*}
\bar{J}(x_0) = \liminf_{T \rightarrow \infty} \frac{1}{1-\alpha}
\sum_{k=1}^n  \int_0^T e^{-\rho T} x^{1-\alpha}_k(t) dt,
\end{eqnarray*}
to be maximized over the consecutive moments of sending congestion
notifications
$T_i,i=1,2,\dots$. In the meanwhile, due to the limited capacities
of the links, the expression  $ \liminf_{T \to \infty} \frac{1}{T}
\int_0^T A x(t) dt$ (resp., $ \liminf_{T \to \infty} \int_0^T
e^{-\rho t} A x(t) dt$) under the average (resp., discounted)
criterion should not be too big. Therefore, after introducing the
weight coefficients $\lambda_1,\dots,\lambda_L\ge 0,$ we consider
the following objective functions to be maximized:
\begin{eqnarray}\label{eq:laverage}
&&\bar{L}(x_1,\dots,x_n) = \sum_{k=1}^n \left\{\liminf_{T \to
\infty} \frac{1}{T}\int_0^T \frac{x^{1-\alpha}_k(t)}{1-\alpha}
dt\right\} -\sum_{l=1}^L \lambda_l  \sum_{k:l\in P(k)} \liminf_{T
\to \infty} \frac{1}{T}\int_0^T x_k(t) dt
\end{eqnarray}
in the average case, and
\begin{eqnarray}\label{eq:ldiscounted}
\bar{L}_\rho(x_1,\dots,x_n) = \sum_{k=1}^n \left\{\liminf_{T \to
\infty} \int_0^T e^{-\rho t} \frac{x^{1-\alpha}_k(t)}{1-\alpha}
dt\right\}-\sum_{l=1}^L \lambda_l  \sum_{k:l\in P(k)} \liminf_{T
\to \infty} \int_0^T e ^{-\rho t}x_k(t) dt
\end{eqnarray}
in the discounted case, where we recall that $P(k)$ indicates the
set of links corresponding to connection $k$.
We can interpret the second terms in (\ref{eq:laverage}) and (\ref{eq:ldiscounted})
as ``soft'' capacity constraints.

Below we obtain the optimal policy for the problems
\begin{eqnarray}\label{ExtendedZ1}
\bar{L}(x_1,\dots,x_n)\rightarrow \max_{T_1,T_2,\dots}.
\end{eqnarray}
and
\begin{eqnarray}\label{DiscountedZ2}
\bar{L}_\rho(x_1,\dots,x_n)\rightarrow \max_{T_1,T_2,\dots},
\end{eqnarray}
respectively.

\subsection{Solving the average optimal impulsive control problem
for the Internet congestion control}
We first consider in this subsection the average problem
(\ref{ExtendedZ1}). Concentrated on policies satisfying
\begin{eqnarray*}&&\liminf_{T
\to \infty} \frac{1}{1-\alpha} \sum_{k=1}^n \frac{1}{T} \int_0^T
x^{1-\alpha}_k(t) dt=\lim_{T \to \infty} \frac{1}{1-\alpha}
\sum_{k=1}^n \frac{1}{T} \int_0^T x^{1-\alpha}_k(t) dt<\infty
\end{eqnarray*}
and $ \liminf_{T \to \infty}\frac{1}{T}\int_0^T x_k(t)dt<\infty$
for each $k=1,\dots,n,$ for problem (\ref{ExtendedZ1}) it is
sufficient to consider the case of $n=1.$ Indeed, one can
legitimately rewrite the function
 {(\ref{eq:laverage})} as
\begin{eqnarray*}
\bar{L} (x_1,\dots,x_n) = \sum_{k=1}^n \liminf_{T \to
\infty}\frac{1}{T}\int_0^T
\left(\frac{x^{1-\alpha}_k(t)}{1-\alpha} - \lambda^k x_k(t)
\right) dt,
\end{eqnarray*} where $
\lambda^k=\sum_{l \in P(k)} \lambda_l,$ which allows us to
decouple different sources. Thus, we will focus on the case of
$n=1$, and solve the following optimal control problem
\begin{equation}\label{eq:laverage1}
\tilde{J}(x_0) = \liminf_{T \to \infty} \frac{1}{T} \int_0^T
\left(\frac{x^{1-\alpha}(t)}{1-\alpha}-\lambda x(t)\right)
dt\rightarrow \max_{T_1,T_2,\dots},
\end{equation}
where $x(t)$ is subject to (\ref{Kostia1}), (\ref{Kostia2}) and
the impulsive controls $T_1,T_2,\dots$ with the initial condition
$x(0)=x_0.$ Here and below the index $k=1$ has been omitted for
convenience.

In the remaining part of this subsection, using the verification
theorem (see Theorem \ref{ZYt2}), we rigorously obtain the optimal
policy and value to problem (\ref{eq:laverage1}) in closed-forms.

Let us start with formulating the congestion control problem
(\ref{eq:laverage1}) in the framework given in the previous
section, which also applies to the next subsection. Indeed, one
can take the following system parameters; $X=(0,\infty),$
$j(x,v)=b^v x$, $C(x,v)=0$ with $v\in V=\{1,2,\dots\},$ $f(x,u)=a
x^\gamma,$ and $c(x,u)=\frac{x^{1-\alpha}}{1-\alpha}-\lambda x$
with $u\in U,$ which is a singleton, i.e., there is no gradual
control, so that in what follows, we omit $u\in U$ everywhere. For
practical reasons, it is reasonable to focus only on policies
$\pi$, under which there is some constant $T^\pi\ge 0$ such that
for each $t>T^\pi,$ $x^\pi(t)$ belongs to a $\pi$-dependent but
$t$-independent compact subset of $X.$ \footnote{This requirement
can be withdrawn in the next subsection dealing with the
discounted problem.}

\begin{theorem} \label{thm:thres}
Suppose $\lambda> 0,$
 {$\gamma\in[0,1],~\alpha>0$, $\alpha\ne 1,$
$2-\alpha-\gamma\ne 0$, $a\in(0,\infty),$ and $b\in (0,1).$} Let
us consider the average congestion control problem
$(\ref{eq:laverage1}).$ Then the optimal policy is given by
$\pi^\ast=({\cal L^\ast},v^{{{\cal L^\ast}}})$ with ${\cal
L^\ast}=[\overline{x},\infty)$, and $v^{{\cal L^\ast}}(x)=k$ if
$x\in [\frac{\overline{x}}{b^{k-1}},
\frac{\overline{x}}{b^k})\subseteq {\cal L^\ast}$ for
$k=1,2,\dots,$ where
\begin{equation}\label{eq:threshold}
\bar{x} =  {\left\{\frac{(2-\gamma)(1-b^{2-\alpha-\gamma})}
{(2-\alpha-\gamma)(1-b^{2-\gamma})\lambda}\right\}}^\frac{1}{\alpha}>0.
\end{equation}
When $\gamma<1,$ the value function is given by
\begin{equation}\label{eq:g}
J({x}_0,\pi^\ast)=g:=\overline{x}\lambda\frac{\alpha}{1-\alpha}\frac{(1-\gamma)(1-b^{2-\gamma})}{(2-\gamma)(1-b^{1-\gamma})};
\end{equation}
and when  $\gamma=1$,
\begin{equation}\label{eq:g1}
J({x}_0,\pi^\ast)=g:=\overline{x}\lambda\frac{\alpha}{1-\alpha}\frac{b-1}{\ln{(b)}}.
\end{equation}
\end{theorem}
\begin{proof} Suppose $\gamma<1$.
By Theorem \ref{ZYt2}, it suffices to show that Condition \ref{ZYCA} is
satisfied by the policy $\pi^\ast=({\cal L^\ast},v^{^{{\cal
L^\ast}}})$, the constant $g$ given by (\ref{eq:g}) and the
function
\begin{equation}\label{eq:h}
h(x)=\left\{\begin{array}{ll}
h_0(x), & \mbox{if} \ x \in (0,\bar{x}),\\
h_k(x)=h_0(b^k x), & \mbox{if} \ x \in
[\bar{x}/b^{k-1},\bar{x}/b^k),
\end{array}\right.
\end{equation}
where $ h_0(x) = \frac{1}{a} \left[
-\frac{x^{2-\alpha-\gamma}}{(1-\alpha)(2-\alpha-\gamma)} +\lambda
\frac{x^{2-\gamma}}{2-\gamma}+g \frac{x^{1-\gamma}}{1-\gamma}
\right], $ and $\bar{x}$ is given by (\ref{eq:threshold}). For
reference and to improve the readability, we write down the
Bellman equation (\ref{ZYBellmanAverage}) for problem
$(\ref{eq:laverage1})$ as follows;
\begin{eqnarray}\label{eq:dpaverageTCPpower}
&&\max\left\{\left(\frac{x^{1-\alpha}}{1-\alpha}-\lambda
x\right)-g + \frac{\partial h}{\partial
x}(x)ax^\gamma,~\sup_{m=1,2,\dots}\{h(b^m x)-h(x)\}\right\}=0.
\end{eqnarray}
Since parts (i,iv) of Condition \ref{ZYCA} are trivially
satisfied, we only verify its parts (ii,iii) as follows.

Consider firstly $x\in(0,\overline{x})=X\setminus {\cal L^\ast}.$
Then, we obtain from direct calculations that
$\left(\frac{x^{1-\alpha}}{1-\alpha}-\lambda x\right)-g +
\frac{\partial h(x)}{\partial x} ax^\gamma=
\left(\frac{x^{1-\alpha}}{1-\alpha}-\lambda x\right)-g +
\frac{\partial h_0(x)}{\partial x} ax^\gamma = 0.$ Let us show
that $\sup_{m=1,2,\dots}\{h(b^m
x)-h(x)\}=\sup_{m=1,2,\dots}\{h_0(b^m x)-h_0(x)\}\le 0$ for
$x\in(0,\overline{x})$ as follows. Define $ \Delta_1(x) := h_0(bx)
- h_0(x) $ for each $x\in (0,\overline{x}).$ Then one can show
that $\Delta_1(x)<0$ for each $b\in(0,1).$ Indeed, direct
calculations give
\begin{eqnarray*}
\Delta_1(x)=\frac{x^{1-\gamma}}{a} \left[
-\frac{(b^{2-\alpha-\gamma}-1)x^{1-\alpha}}{(1-\alpha)(2-\alpha-\gamma)}
+\lambda \frac{(b^{2-\gamma}-1)x}{2-\gamma} +g
\frac{(b^{1-\gamma}-1)}{1-\gamma}\right],
\end{eqnarray*}
so that for the strict negativity of $\Delta_1(x)$, it is
equivalent to showing it for the following expression
\begin{eqnarray*}
\tilde \Delta_1(x) :=
-\frac{(b^{2-\alpha-\gamma}-1)x^{1-\alpha}}{(1-\alpha)(2-\alpha-\gamma)}
+\lambda \frac{(b^{2-\gamma}-1)x}{2-\gamma} +g
\frac{(b^{1-\gamma}-1)}{1-\gamma},
\end{eqnarray*}
whose first order and second order derivatives (with respect to
$x$) are given by
\begin{eqnarray*}
\tilde \Delta_1'(x) =
-\frac{(b^{2-\alpha-\gamma}-1)x^{-\alpha}}{2-\alpha-\gamma}
+\lambda \frac{(b^{2-\gamma}-1)}{2-\gamma}
\end{eqnarray*}
and
\begin{eqnarray*}
\tilde \Delta_1''(x) =
\frac{\alpha(b^{2-\alpha-\gamma}-1)x^{-\alpha-1}}{2-\alpha-\gamma}.
\end{eqnarray*}
Under the conditions of the parameters, $\tilde
\Delta_1''(x)<0$ for each $x\in(0,\overline{x}),$ and thus the
function $\tilde \Delta_1(x) $ is concave on $(0,\overline{x})$
achieving its unique maximum at the stationary point given by $
 x=\overline{x}=  \left\{{\frac{(2-\gamma)(1-b^{2-\alpha-\gamma})}
{(2-\alpha-\gamma)(1-b^{2-\gamma})\lambda}}\right\}^{\frac{1}{\alpha}}>0.
$ Note that $\tilde{\Delta}_1(\overline{x})=0$ and
$\lim_{x\downarrow 0}\tilde{\Delta}_1(x) \le 0.$ It follows from
the above observations and the standard analysis of derivatives
that $\tilde \Delta_1(x)<0$ and thus $\Delta_1(x)<0$ for each
$x\in(0,\overline{x}).$ Since $
\frac{\partial\overline{x}}{\partial b}\le 0 $ for each
$b\in(0,1)$ as can be easily verified, one can replace $b$ with
$b^m$ ($m=2,3,\dots$) in the above argument to obtain that
$h_0(b^mx) - h_0(x)<0$ for each $x\in(0,\overline{x})$, and thus
\begin{eqnarray}\label{ZYUseful1}
\sup_{m=1,2,\dots}\{h_0(b^m x)-h_0(x)\}\le 0
\end{eqnarray}
for $x\in(0,\overline{x}),$ as desired. Hence, it follows that
Condition \ref{ZYCA}(ii,iii) is satisfied on $(0,\overline{x}).$

Next, we show by induction that Condition \ref{ZYCA}(ii,iii) is
satisfied on
$[\frac{\overline{x}}{b^{k-1}},\frac{\overline{x}}{b^k})$,
$k=1,2,\dots.$ Let us consider the case of $k=1$, i.e., the
interval $[\bar{x},\frac{\bar{x}}{b})$. By the definition of the
function $h(x)$, we have
\begin{eqnarray}\label{ZYV0}
\sup_{m=1,2,\dots}\{h(b^mx)-h(x)\}= 0
\end{eqnarray}
for $x \in [\bar{x},\bar{x}/b)$. Indeed, by the definition of
$h(x)$, we have
\begin{eqnarray}\label{ZYN}
h(bx)-h(x)=0
\end{eqnarray}
for $x \in [\bar{x},\bar{x}/b)$, whereas for each $m=2,3,\dots$
and $x\in [\bar{x},\frac{\bar{x}}{b}),$ it holds that
$h(b^mx)-h(x)=h_0(b^m x)-h_0(bx)\le 0$, which follows from that
$bx\in(0,\overline{x})$, $b^mx=b^{m-1}(bx)\in(0,\overline{x})$ for
each $x\in [\bar{x},\frac{\bar{x}}{b}),$ and (\ref{ZYUseful1}).
Furthermore, one can show that
\begin{eqnarray}\label{ZYV1}
\Delta_2(x) &:=& \frac{x^{1-\alpha}}{1-\alpha}-\lambda x - g +
\frac{\partial h(x)}{\partial x} a x^\gamma\le0
\end{eqnarray}
for each $x \in [\bar{x},\bar{x}/b),$ which follows from the
following observations. Since $h(x)=h_1(x)=h_0(bx),$ we see $
\Delta_2(x)=-\frac{(b^{2-\alpha-\gamma}-1)x^{1-\alpha}}{1-\alpha}
+\lambda (b^{2-\gamma}-1) x + g (b^{1-\gamma}-1)$ for each $x\in
[\bar{x},\frac{\bar{x}}{b}),$ and in particular,
\begin{eqnarray}\label{ZYSaturday}
\Delta_2(\overline{x})=0,
\end{eqnarray}
as can be easily verified. The derivative of the function
$\Delta_2(x)$ with respect to $x$ is given by $
\Delta_2'(x)=-(b^{2-\alpha-\gamma}-1)x^{-\alpha}+\lambda
(b^{2-\gamma}-1). $ If $2-\gamma-\alpha<0,$ then $\Delta_2'(x)<0$,
which together with (\ref{ZYSaturday}) shows $\Delta_2(x)\le 0$ on
$[\bar{x},\frac{\bar{x}}{b}).$ If $2-\gamma-\alpha>0$, then
$\Delta_2''(x)= \alpha(b^{2-\alpha-\gamma}-1)x^{-\alpha-1}<0$ and
thus, the function $\Delta_2(x)$ is concave with the maximum
attained at the stationary point
$x=\left(\frac{1-b^{2-\alpha-\gamma}}{(1-b^{2-\gamma})\lambda}\right)^{\frac{1}{\alpha}}.$
Since
$\left(\frac{1-b^{2-\alpha-\gamma}}{(1-b^{2-\gamma})\lambda}\right)^{\frac{1}{\alpha}}\le
\overline{x}$, (\ref{ZYSaturday}) implies $\Delta_2(x)\le 0$ on
$[\bar{x},\frac{\bar{x}}{b}),$ as desired. By the way, for the
later reference, the above observations actually show that
\begin{eqnarray}\label{ZZZZZZ}
G(x)&:=&-\frac{(b^{2-\alpha-\gamma}-1)x^{1-\alpha}}{1-\alpha}
+\lambda (b^{2-\gamma}-1) x + g (b^{1-\gamma}-1)\le 0
\end{eqnarray}
for all $x \ge \overline{x}.$ Thus, combining (\ref{ZYV0}), (\ref{ZYN}), 
and (\ref{ZYV1}) shows that Condition \ref{ZYCA}(ii,iii) is
satisfied on $[\bar{x},\frac{\bar{x}}{b}).$

Assume that for each
$x\in[\frac{\overline{x}}{b^{k-1}},\frac{\overline{x}}{b^k})$ and
each $k=1,2,\dots,M,$ relations (\ref{ZYV0}) and (\ref{ZYV1})
hold, together with
\begin{eqnarray}\label{ZYYY}
h(b^{k}x)-h(x)=0 \mbox{~(the corresponding version of
(\ref{ZYN}))}.
\end{eqnarray}
Now we consider the case of $k=M+1,$ i.e., when $x\in
[\frac{\overline{x}}{b^{M}},\frac{\overline{x}}{b^{M+1}})$. For
each $x\in
[\frac{\overline{x}}{b^{M}},\frac{\overline{x}}{b^{M+1}}),$ when
$m=1,2,\dots,M,$ it holds that $b^m x\in[
\frac{\overline{x}}{b^{M-m}},\frac{\overline{x}}{b^{M+1-m}})$, and
thus $ h(b^m x)-h(x)=h_0(b^{M+1}x)-h_0(b^{M+1}(x))=0; $ when
$m=M+1,M+2,\dots$, $b^m x\in (0,\overline{x})={\cal L^\ast}$, and
thus $ h(b^m x)-h(x)=h_0(b^m x)-h_0(b^{M+1}x)=0$ if $m=M+1$, and $
h(b^m x)-h(x)=h_0(b^{m-(M+1)}(b^{M+1} x))-h_0(b^{M+1}x)\le 0$ if
$m>M+1,$ by (\ref{ZYUseful1}). Thus, we see (\ref{ZYV0}) holds for
$x\in [\frac{\overline{x}}{b^{M}},\frac{\overline{x}}{b^{M+1}})$.
Note that in the above we have also incidentally verified the
validity of (\ref{ZYYY}) for the case of $k=M+1$.

Below we verify (\ref{ZYV1}) for the case of $k=M+1,$ which would
complete the proof by induction. To this end, we first present
some preliminary observations that hold for each $k=1,2,\dots.$
For each $k=1,2,\dots,$ since $h(x)=h_k(x)=h_0(b^kx)$ for each $x
\in [\frac{\overline{x}}{b^{k-1}},\frac{\overline{x}}{b^{k}})$, we
have
\begin{eqnarray*}
\Delta_2(x):=-\frac{(b^{k
(2-\alpha-\gamma)}-1)x^{1-\alpha}}{(1-\alpha)} +\lambda x
{(b^{k(2-\gamma)}-1)}+g {(b^ {k(1-\gamma)}-1)}.
\end{eqnarray*}
For the convenience of later reference, let us introduce the
notation
\begin{eqnarray*}
\tilde{\Delta}_k(x)&:=&-\frac{(b^{k(2-\alpha-\gamma)}-1)x^{1-\alpha}}{1-\alpha}
+\lambda (b^{k(2-\gamma)}-1) x + g (b^{k(1-\gamma)}-1)\\
&=&b^{k(1-\gamma)}(-\frac{b^{k(1-\alpha)}x^{1-\alpha}}{1-\alpha}
+\lambda b^{k} x + g )-(-\frac{x^{1-\alpha}}{1-\alpha} +\lambda  x
+ g )
\end{eqnarray*}
for each $x>0.$ Therefore, for $x \in
[\frac{\overline{x}}{b^{k-2}},\frac{\overline{x}}{b^{k-1}})$, we
have
\begin{eqnarray*}
\Delta_2(x)=\tilde{\Delta}_{k-1}(x)=b^{(k-1)(1-\gamma)}(-\frac{b^{(k-1)(1-\alpha)}x^{1-\alpha}}{1-\alpha}
+\lambda b^{k-1} x + g )-(-\frac{x^{1-\alpha}}{1-\alpha} +\lambda
x + g).
\end{eqnarray*}
Let us define
\begin{eqnarray*}
F(x):=-\frac{x^{1-\alpha}}{1-\alpha}+\lambda x + g
\end{eqnarray*}
for each $x>0.$  We then have from the direct calculations that
\begin{eqnarray}\label{ZZZZZZ1}
\tilde{\Delta}_{k-1}(\frac{\bar{x}}{b^{k-2}})&=&b^{(k-1)(1-\gamma)}
F(b\bar{x})-F(\frac{\bar{x}}{b^{k-2}})
\end{eqnarray}
for each $k=1,2,\dots.$ Focusing on $ F(\frac{\bar{x}}{b^{k-2}}),$
we have
\begin{eqnarray*}
b^{1-\gamma}F(\frac{\bar{x}}{b^{k-2}})&=&
-\frac{\bar{x}^{1-\alpha}}{1-\alpha}
\frac{b^{1-\gamma}}{b^{(k-2)(1-\alpha)}}+\lambda
\bar{x}\frac{b^{1-\gamma}}{b^{k-2}} +  g b^{1-\gamma}\\
& =&  -\frac{\bar{x}^{1-\alpha}}{1-\alpha}
\frac{b^{2-\alpha-\gamma}}{b^{(k-1)(1-\alpha)}}+\lambda
\bar{x}\frac{b^{2-\gamma}}{b^{k-1}} + g b^{1-\gamma}\\
& =&
-\frac{(\frac{\bar{x}}{b^{(k-1)}})^{1-\alpha}}{1-\alpha}b^{2-\alpha-\gamma}
+\lambda(\frac{ \bar{x}}{b^{k-1}}) b^{2-\gamma}+ g
  b^{1-\gamma}.
\end{eqnarray*}
Recall that in the above, we have proved that $G(x)\leq 0$ for
$x\geq\bar{x}$, see (\ref{ZZZZZZ}). Thus, we have $ G(\frac{
\bar{x}}{b^{k-1}})\leq0, $ i.e., $
-\frac{b^{2-\alpha-\gamma}(\frac{
\bar{x}}{b^{k-1}})^{1-\alpha}}{1-\alpha}+\lambda b^{2-\gamma}
(\frac{ \bar{x}}{b^{k-1}}) + g b^{1-\gamma}\leq-\frac{(\frac{
\bar{x}}{b^{k-1}})^{1-\alpha}}{1-\alpha}+\lambda (\frac{
\bar{x}}{b^{k-1}}) + g.$ Consequently,
\begin{eqnarray*}
b^{1-\gamma}F(\frac{\bar{x}}{b^{k-2}})\leq  -\frac{(\frac{
\bar{x}}{b^{k-1}})^{1-\alpha}}{1-\alpha}+\lambda (\frac{
\bar{x}}{b^{k-1}}) + g=F(\frac{\bar{x}}{b^{k-1}}).
\end{eqnarray*}

Now we verify (\ref{ZYV1}) for the particular case of $k=M+1$. By
the inductive supposition, (\ref{ZYV1}) holds for
$x\in[\frac{\overline{x}}{b^{M-1}},\frac{\overline{x}}{b^M})$, we
thus have $\Delta_2(\frac{\bar{x}}{b^{M-1}})\leq0$, and
\begin{eqnarray*}
0&\geq&\tilde{\Delta}_{M}(\frac{\bar{x}}{b^{M-1}})=b^{M(1-\gamma)}
F(b\bar{x})-F(\frac{\bar{x}}{b^{M-1}})\geq b^{(M)(1-\gamma)}
F(b\bar{x})-\frac{1}{b^{1-\gamma}}F(\frac{\bar{x}}{b^{M}}).
\end{eqnarray*}
Therefore, we obtain that  $b^{(M+1)(1-\gamma)}
F(b\bar{x})-F(\frac{\bar{x}}{b^{M}})\leq0,$ and by
(\ref{ZZZZZZ1}),
\begin{eqnarray}\label{ZZZstar}
\Delta_2(\frac{\bar{x}}{b^{M}})\leq 0.
\end{eqnarray}
Furthermore, the derivative of the function
$\tilde{\Delta}_{M+1}(x)$ with respect to $x$ is given by $
\tilde{\Delta}_{M+1}'(x)=-(b^{(M+1)(2-\alpha-\gamma)}-1)x^{-\alpha}+\lambda
(b^{(M+1)(2-\gamma)}-1). $ If $2-\gamma-\alpha<0,$ then
$\tilde{\Delta}_{M+1}'(x)<0$. Thus, by (\ref{ZZZstar}), we obtain
that $ \Delta_2(x)=\tilde{\Delta}_{M+1}(x) \leq0 $ for $x \in
[\frac{\overline{x}}{b^{M}},\frac{\overline{x}}{b^{M+1}})$. If
$2-\gamma-\alpha>0$, then $\tilde{\Delta}_{M+1}''(x)=
\alpha(b^{(M+1)(2-\alpha-\gamma)}-1)x^{-\alpha-1}<0$, and in turn,
the function $\tilde{\Delta}_{M+1}(x)$ is concave with the maximum
attained at the stationary point $
x=\left(\frac{1-b^{(M+1)(2-\alpha-\gamma)}}{(1-b^{(M+1)(2-\gamma)})\lambda}\right)^{\frac{1}{\alpha}}.
$ Moreover, we have $
\frac{\sum_{m=0}^{M}b^{m(2-\alpha-\gamma)}}{\sum_{m=0}^{M}b^{m(2-\gamma)}}\leq\frac{1}{b^{M\alpha}},$
which follows from the fact that for each $m=0,1,\dots,M,$
$m(2-\alpha-\gamma)+M\alpha\geq m(2-\gamma),$  so that
$b^{m(2-\alpha-\gamma)}b^{M\alpha}\leq b^{m(2-\gamma)}$. From this
we see
\begin{eqnarray*}
&&\frac{(1-b^{2-\alpha-\gamma})\sum_{m=0}^{M}b^{m(2-\alpha-\gamma)}}{(1-b^{2-\gamma})\sum_{m=0}^{M}b^{m(2-\gamma)}}\leq\frac{1}{b^{M\alpha}}\frac{1-b^{2-\alpha-\gamma}}{1-b^{2-\gamma}}\\
&\Rightarrow&\frac{1-b^{(M+1)(2-\alpha-\gamma)}}{(1-b^{(M+1)(2-\gamma)})\lambda}\leq\frac{1}{b^{M\alpha}}\frac{(2-\gamma)(1-b^{2-\alpha-\gamma})}{(2-\alpha-\gamma)(1-b^{2-\gamma})\lambda}\\
&\Leftrightarrow&
\left(\frac{1-b^{(M+1)(2-\alpha-\gamma)}}{(1-b^{(M+1)(2-\gamma)})\lambda}\right)^{(\frac{1}{\alpha})}\leq
\frac{\bar{x}}{b^{M}}.
\end{eqnarray*}  Finally, it follows from
the last line of the previous inequalities, the concavity of the
function $\tilde{\Delta}_{M+1}$ and (\ref{ZZZstar}) that $
\Delta_2(x)\leq 0 $ for $x \in
[\frac{\overline{x}}{b^{M}},\frac{\overline{x}}{b^{M+1}})$, which
verifies (\ref{ZYV1}), and thus completes the proof.

The case of $\gamma=1$ can be similarly treated.
\end{proof}

\subsection{Solving the discounted optimal impulsive control problem
for the Internet congestion control}
The discounted problem turns out more difficult to deal with, and
we suppose the sending rate increases additively, i.e., $\frac{d
x_k(t)}{dt}=a_k>0$, and decreases multiplicatively, i.e.,
$j(x,v)=bx$ with $b\in (0,1)$ when a congestion notification is sent, see
(\ref{Kostia1}) and (\ref{Kostia2}). Furthermore, we assume
$\alpha\in (1,2).$

Similarly to the average case, upon rewriting the objective
function in problem (\ref{DiscountedZ2}) as
$\bar{L}_\rho(x_1,\dots,x_n) = \sum_{k=1}^n \liminf_{T \to \infty}
\int_0^T e^{-\rho t} \left(\frac{x^{1-\alpha}_k(t)}{1-\alpha} -
\lambda^k x_k(t) \right) dt,$  where $ \lambda^k=\sum_{l \in P(k)}
\lambda_l,$ it becomes clear that there is no loss of generality
to focus on the case of $n=1$;
\begin{equation}\label{eq:ldiscounte1}
\tilde{J}_\rho(x_0) = \liminf_{T \to \infty} \int_0^T e^{-\rho t}
\left(\frac{x^{1-\alpha}(t)}{1-\alpha}-\lambda x(t)\right)
dt\rightarrow \max_{T_1,T_2,\dots},
\end{equation}

Now the Bellman equation (\ref{e5}) has the form
  \begin{equation}\label{e6}
\max\left\{\frac{x^{1-\alpha}}{1-\alpha}-\lambda x-\rho
W(x)+a\frac{dW}{dx},~~\sup_{i\ge 1}[W(b^ix)-W(x)]\right\}=0.
  \end{equation}
The linear differential equation
  \begin{equation}\label{e6p}
\frac{x^{1-\alpha}}{1-\alpha}-\lambda x-\rho\tilde
W(x)+a\frac{d\tilde W}{dx}=0
  \end{equation}
can be integrated:
  \begin{equation}\label{e7}
\tilde
W(x)=e^{\frac{\rho}{a}(x-1)}\biggl(\frac{\lambda}{\rho}+\frac{\lambda
a}{\rho^{2}}+\frac{1}{\rho(\alpha-1)}+\tilde
w_{1}\biggr)-\frac{x^{1-\alpha}}{\rho(\alpha-1)}-\frac{\lambda}{\rho}x-\frac{a\lambda}{\rho^{2}}-\frac{e^{\frac{\rho}{a}x}}{\rho}\int_{1}^{x}e^{-\frac{\rho}{a}u}u^{-\alpha}du.
\end{equation}
Here $\tilde w_1=\tilde W(1)$ is a fixed parameter.

Suppose for a moment that no impulses are allowed, so that
$x(t)=x_0+at$. We omit the $\pi$ index because here is a single
control policy. We have a family of functions $\tilde W(x)$
depending on the initial value $w_1$, but only one of them
represents the criterion
  $$\liminf_{T\to\infty}\int_0^T\left\{ e^{-\rho t}\frac{(x(t))^{1-\alpha}}{1-\alpha}-\lambda x(t)\right\} dt=W^*(x_0).$$
In this situation, for the function $\tilde W$, all the parts of
Condition \ref{c1} are obviously satisfied (${\cal  D}=\emptyset,
T_1=\infty, {\cal L}^*=\emptyset$) except for (iv).

Since $W^*<0$, the case $\limsup_{T\to\infty} e^{-\rho T}
W^*(x(T))>0$ is excluded and we need to find such an initial value
$w_1^*$ that
  \begin{equation} \label{e8}
\lim_{T\to\infty} e^{-\rho T} \tilde W(x(T))=0,~~~\mbox{ where
}~~~x(T)=x_0+aT,~~x_0>0.
  \end{equation}
Equation (\ref{e8}) is equivalent to the following:
  $$\lim_{T\to\infty} e^{\frac{\rho}{a}x_0}\left\{ e^{-\frac{\rho}{a}}\left(\frac{\lambda}{\rho}+\frac{\lambda a}{\rho^2}+\frac{1}{\rho(\alpha-1)}+w_1\right)-\frac{1}{\rho}\int_1^{x_0+aT} e^{-\frac{\rho}{a} u} u^{-\alpha} du\right\}=0.$$
Therefore,
  \begin{equation}\label{e9}
w_1^*=\frac{e^{\frac{\rho}{a}}}{\rho}\left(\frac{\rho}{a}\right)^{\alpha-1}\Gamma\left(1-\alpha,\frac{\rho}{a}\right)-\frac{1}{\rho(\alpha-1)}-\frac{\lambda(\rho+a)}{\rho^2},
  \end{equation}
and $W^*(x_0)$ is given by (\ref{e7}) at $w_1=w_1^*$. Here
$\Gamma(y,z)=\int_z^\infty e^{-u} u^{y-1} du$ is the incomplete
gamma function \cite[3.381-3]{b1}. By the way, $W^*$ is the
maximal non-negative solution to the differential equation
(\ref{e6p}).

For the discounted impulsive control problem (\ref{DiscountedZ2}),
the solution is given in the following statement.

\begin{theorem}\label{t2} \par\noindent(a) Equation
  \begin{eqnarray}
H(x)&:=&\left(e^{\frac{\rho x(1-b)}{a}}-1\right)\frac{(1-b)\lambda a}{\rho}-(1-b)e^{\frac{\rho x}{a}}\int_{bx}^x e^{-\frac{\rho u}{a}} u^{-\alpha} du \label{e10}\\
&& -\left(e^{\frac{\rho
x(1-b)}{a}}-b\right)\left[\frac{x^{1-\alpha}(1-
b^{1-\alpha})}{\alpha-1}+\lambda x(1-b)\right] \nonumber =0
\nonumber
  \end{eqnarray}
has a single positive solution $\bar x$.

\par\noindent(b) Let
  \begin{eqnarray}
w_1&=& \frac{e^{\frac{\rho}{a}}}{\rho}\int_1^{\bar x} e^{-\frac{\rho u}{a}} u^{-\alpha} du -\frac{\lambda(\rho+a)}{\rho^2}-\frac{1}{\rho(\alpha-1)}\label{e11p}\\
&&-\left.\left[\frac{\bar x^{1-\alpha}(1-
b^{1-\alpha})}{\rho(\alpha-1)}+\frac{(1-b)\lambda \bar
x}{\rho}+\frac{e^{\frac{\rho b \bar x}{a}}}{\rho}\int_{b\bar
x}^{\bar x} e^{-\frac{\rho}{a}u} u^{-\alpha} du
\right]\right/\left( e^{\frac{\rho b \bar
x-\rho}{a}}-e^{\frac{\rho \bar x-\rho}{a}}\right)\nonumber
  \end{eqnarray}
and, for $0<x<\bar x$, put $W(x)=\tilde W(x)$, where $\tilde W$ is
given by formula (\ref{e7}) under $\tilde w_1=w_1$. For the
intervals $\left[\bar x, \frac{\bar x}{b}\right)$,
$\left[\frac{\bar x}{b},\frac{\bar x}{b^2}\right)$, $\ldots$ the
function $W$ is defined recursively: $W(x):= W(bx)$. Then the
function $W$ satisfies items (i, ii, iii) of Condition \ref{c1}.

\par\noindent(c) The function $W(x_0)=\sup_\pi {J}_\rho(x_0,\pi)={J}_\rho(x_0,\pi^*)$ is the Bellman
function, where the (feedback) optimal policy $\pi^*$ is given by
  $${\cal L}^*=[\bar x,\infty),~~~v^{f,{\cal L}^*}(x)=i~~~\mbox{ if } x\in\left[\frac{\bar x}{b^{i-1}}, \frac{\bar x}{b^i}\right).$$
\end{theorem}

Some comments and remarks are in position, before we give the
proof of this theorem. For $b=0.5$, $\rho=1$, $\alpha=1.3$,
$\lambda=2$, $a=0.2$ the graph of function $W$ is presented on
Fig.\ref{figure1}. Here $\bar x=0.7901$ and $w_1=-4.9301$. The
dashed line represents the graph of function
\begin{equation}\label{e9p}
z(x)=-\frac{1}{\rho}\left(\frac{x^{1-\alpha}}{\alpha-1}+\lambda
x\right).
\end{equation}
When $\tilde W(x)=z(x)$, we have $\frac{d\tilde W}{dx}=0$; if
$\tilde W(x)>z(x)$ ($\tilde W(x)<z(x)$) function $\tilde W$
increases (decreases). The dotted line represents the graph of
function
  $$v(x)=\frac{a(x^{-\alpha}-\lambda)}{\rho^2}-\frac{1}{\rho}\left(\frac{x^{1-\alpha}}{\alpha-1}+\lambda x\right).$$
If $\tilde W(x)=v(x)$ then from (\ref{e6p}) we have
\begin{eqnarray*}
a^2\frac{d^2\tilde W}{dx^2}&=&a^2\left[ a\rho\frac{d\tilde W}{dx}+\lambda a-a x^{-\alpha}\right]\\
&=&a^2\left[\rho^2\tilde W(x)+\rho\lambda x+\frac{\rho
x^{1-\alpha}}{\alpha-1}+\lambda a-ax^{-\alpha}\right]=0,
\end{eqnarray*}
that is, $x$ is the point of inflection of function $\tilde W$.
This reasoning applies to any solution of equation (\ref{e6p}).

\begin{figure}[htbp]
\begin{center}
\includegraphics[width=12cm]{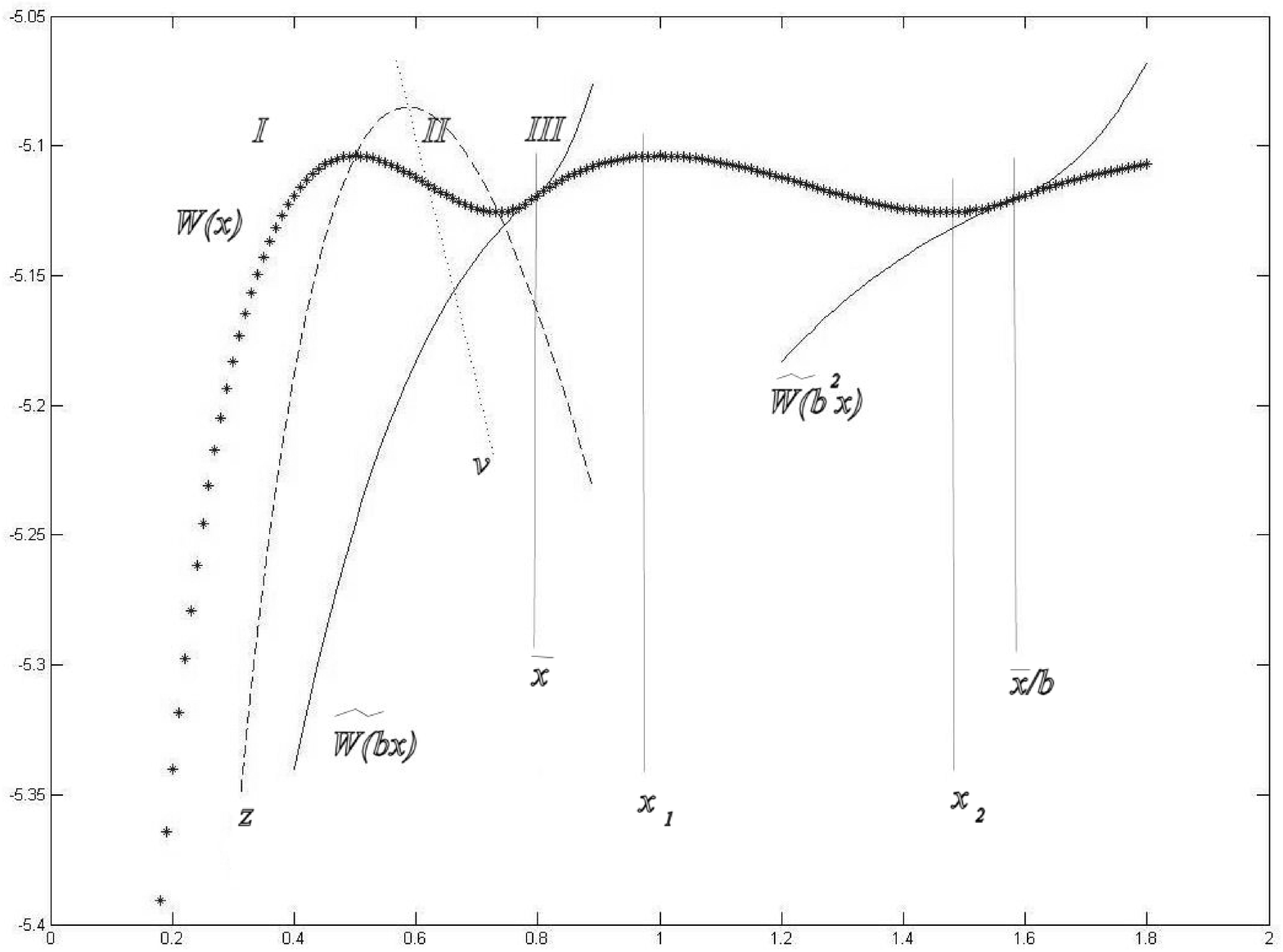}
\caption{Graph of the Bellman function $W(x)$ (bright line with
the star point markers).}\label{figure1}
\end{center}
\end{figure}

On the graph, for $0<x<\bar x$, the Bellman function $W(x)=\tilde
W(x)$ has three parts, denoted below as I,II and III, where it
increases, strictly decreases, and again increases.
Correspondingly, function $\tilde W(bx)$ also has three parts I,II
and III where it increases, strictly decreases and increases
again, and $W(x)=\tilde W(bx)$ for $\bar x\le x<\frac{\bar x}{b}$.
Point $\bar x$ is such that
  \begin{equation}\label{e12}
  \tilde W(\bar x)=\tilde W(b\bar x) ~~\mbox{ and }~~\left.\frac{d\tilde W(x)}{dx}\right|_{\bar x}=\left.\frac{d\tilde W(bx)}{dx}\right|_{\bar x}.
  \end{equation}
As is shown in the proof of Theorem \ref{t2}, these two equations
are satisfied if and only if $\bar x$ solves equation (\ref{e10}).

Let us calculate the limit of $\bar x$ when $\rho$ approaches
zero. One can easily show that, for any $x>0$,
  $$\lim_{\rho\to 0} H(x)=0~~\mbox{ and }~~ \lim_{\rho\to 0}\frac{H(x)}{\rho}=\frac{x^2(1-b)}{a}\left[\frac{\lambda(b^2-1)}{2}+\frac{x^{-\alpha}(1-b^{2-\alpha})}{2-\alpha}\right].$$
Let
  \begin{equation}\label{e11pp}
\bar
x_0=\left[\frac{2(1-b^{2-\alpha})}{\lambda(1-b^2)(2-\alpha)}\right]^{1/\alpha},
  \end{equation}
i.e.
  $$\lim_{\rho\to 0}\frac{H(x)}{\rho}\left\{\begin{array}{cl}
>0, & \mbox{ if } x<\bar x_0,\\<0, & \mbox{ if } x>\bar x_0,\\=0, & \mbox{ if } x=\bar x_0.
  \end{array}\right.$$
Function $\frac{H(x)}{\rho}$ is continuous wrt $\rho$. Therefore,
for any small enough $\varepsilon>0$,
  $$\exists\delta>0:~\forall\rho\in(0,\delta)~~\frac{H(\bar x_0-\varepsilon)}{\rho}>0~~\mbox{ and }~~\frac{H(\bar x_0+\varepsilon)}{\rho}<0$$
meaning that $\bar x_\rho$, the solution to (\ref{e10}) at
$\rho\in(0,\delta)$, satisfies $\bar x_\rho\in(\bar
x_0-\varepsilon,\bar x_0+\varepsilon)$. This means $\lim_{\rho\to
0+}\bar x_\rho=\bar x_0$. Note that (\ref{e11pp}) is the optimal
threshold if we consider the long-run average reward with the same
reward rate $c(x)$.
\bigskip
\par\noindent\textit{Proof of Theorem \ref{t2}.}
\begin{proof}
(a) Firstly, let us prove that no more than one positive number
$\bar x$ can satisfy equations (\ref{e12}). If $\bar x$ satisfies
(\ref{e12}) then function $\tilde W$ canot have only one
increasing branch above function $v$ because two increasing
functions $\tilde W(x)$ and $\tilde W(bx)$ cannot have common
points.

The increasing part I of function $\tilde W(x)$ cannot intersect
with $\tilde W(bx)$.

The strictly decreasing part II of function $\tilde W(x)$ cannot
intersect with the parts II and III of function $\tilde W(bx)$.
Possible common points with the part I  of $\tilde W(bx)$ are of
no interest because here $\frac{d\tilde W(x)}{dx}<0$ and
$\frac{d\tilde W(bx)}{dx}\ge 0$.

The increasing part III of function $\tilde W(x)$ can intersect
with the parts I and II of function $\tilde W(bx)$, but again the
latter case is of no interest because here $\frac{d\tilde
W(x)}{dx}\ge 0$ and $\frac{d\tilde W(bx)}{dx}<0$.

Thus, the only possibility to satisfy (\ref{e12}) is the case when
the increasing part III of $\tilde W(x)$ touches the increasing
part I of function $\tilde W(bx)$. The inflection line $v(x)$ is
located between the increasing and decreasing branches of the
function $z(x)$, so that the part III of $\tilde W(x)$ is convex
and the part I of $\tilde W(bx)$ is concave, meaning that no more
than one point $\bar x$ can satsify the equations (\ref{e12}).

Using formula (\ref{e7}), the equations (\ref{e12}) can be
rewritten as follows:
\begin{eqnarray}
0&=&\tilde W(x)-\tilde W(bx) = \left( e^{\frac{\rho x}{a}}- e^{\frac{b\rho x}{a}}\right)\left[ e^{-\frac{\rho}{a}}\left(\frac{\lambda}{\rho}+\frac{\lambda a}{\rho^2}+\frac{1}{\rho(\alpha-1)}+\tilde w_1\right)\right. \nonumber \\
&&\left. -\frac{1}{\rho}\int_1^x e^{-\frac{\rho u}{a}} u^{-\alpha} du\right]-\frac{x^{1-\alpha}(1-b^{1-\alpha})}{\rho(\alpha-1)}-(1-b)\frac{\lambda x}{\rho}+\frac{e^{\frac{b\rho x}{a}}}\rho}\int_x^{bx} e^{-\frac{\rho u}{a} u^{-\alpha} du; \label{e13}\\
0&=&\frac{d\tilde W(x)}{dx}-\frac{d\tilde W(bx)}{dx}=\left( \frac{\rho}{a} e^{\frac{\rho x}{a}}- \frac{b\rho}{a}e^{\frac{b\rho x}{a}}\right)\left[ e^{-\frac{\rho}{a}}\left(\frac{\lambda}{\rho}+\frac{\lambda a}{\rho^2}+\frac{1}{\rho(\alpha-1)}+\tilde w_1\right)\right. \nonumber\\
&&\left.-\frac{1}{\rho}\int_1^x e^{-\frac{\rho u}{a}} u^{-\alpha}
du\right]-(1-b)\frac{\lambda}{\rho} +\frac{b e^{\frac{b\rho
x}{a}}}{a}\int_x^{bx}e^{-\frac{\rho u}{a}} u^{-\alpha} du.
\nonumber
\end{eqnarray}
After we multiply these equations by factors
$\left(1-be^{(b-1)\frac{\rho x}{a}}\right)$ and
$\frac{a}{\rho}\left(1-e^{(b-1)\frac{\rho x}{a}}\right)$
correspondingly and subtract the equations, the variable $\tilde
w_1$ is cancelled and we obtain equation
\begin{eqnarray*}
0&=&\left(1-b e^{(b-1)\frac{\rho
x}{a}}\right)\left[\frac{b^{1-\alpha}x^{1-\alpha}}{\rho(\alpha-1)}+\frac{e^{\frac{b\rho
x}{a}}}{\rho}\int_x^{bx} e^{\frac{-\rho u}{a}} u^{-\alpha} du-
\frac{x^{1-\alpha}}{\rho(\alpha-1)}-(1-b)\frac{\lambda x}{\rho}\right]\\
&&-\frac{a}{\rho}\left(1- e^{\frac{(b-1)\rho
x}{a}}\right)\left[\frac{b e^{\frac{b\rho x}{a}}}{a}\int_x^{bx}
e^{-\frac{\rho u}{a}} u^{-\alpha}
du-(1-b)\frac{\lambda}{\rho}\right]
\end{eqnarray*}
which is equivalent to $H(x)=0$.

Equation (\ref{e11p}) follows directly from the first of equations
(\ref{e13}): if we know the value of $x$ (equal $\bar x$), we can
compute the value of $\tilde w_1=w_1$.

To prove the solvability of the equation (\ref{e10}) we compute
the following limits:
\begin{eqnarray*}
\lim_{x\to\infty} H(x)&\le& -\lim_{x\to\infty} e^{\frac{\rho x(1-b)}{a}}\cdot \lambda x(1-b)=-\infty;\\
\lim_{x\to 0} H(x)&=& \lim_{x\to 0}
(b-1)\left[\frac{x^{1-\alpha}(1-b^{1-\alpha})}{\alpha-1}+\int_{bx}^x
e^{-\frac{\rho u}{a}} u^{-\alpha} du \right],
\end{eqnarray*}
and the positive expression in  the square brackets does not
exceed
  $$\frac{x^{1-\alpha}(1-b^{1-\alpha})}{\alpha-1}+\int_{bx}^x\left[1-\frac{\rho u}{a}+\frac{1}{2}\left(\frac{\rho u}{a}\right)^2\right] u^{-\alpha} du=\frac{x^{1-\alpha}(1-b^{1-\alpha})}{\alpha-1}$$
  $$+\left[\frac{u^{1-\alpha}}{1-\alpha}-\frac{\rho u^{2-\alpha}}{a(2-\alpha)}+\frac{\rho^2 u^{3-\alpha}}{2a^2(3-\alpha)}\right]_{bx}^x=\left[\frac{\rho^2 u^{3-\alpha}}{2a^2(3-\alpha)}-
\frac{\rho u^{2-\alpha}}{a(2-\alpha)}\right]_{bx}^x\to 0 \mbox{ as
} x\to 0,$$ so that $\lim_{x\to 0} H(x)=0$.

Finally,
\begin{eqnarray*}
\frac{dH}{dx}&=&-\frac{\rho(1-b)}{a}e^{\frac{\rho x(1-b)}{a}}\left[\frac{x^{1-\alpha}(1-b^{1-\alpha})}{\alpha-1}+\lambda x(1-b)\right]\\
&&+\left(b-e^{\frac{\rho x(1-b)}{a}}\right)\left[\lambda(1-b) -x^{-\alpha}(1-b^{1-\alpha})\right]+\lambda(1-b)^2 e^{\frac{\rho x(1-b)}{a}}\\
&&-\frac{\rho}{a}(1-b)e^{\frac{\rho x}{a}}\int_{bx}^x e^{-\frac{\rho u}{a}} u^{-\alpha} du-(1-b) e^{\frac{\rho x}{a}}\left[ e^{-\frac{\rho x}{a}} x^{-\alpha}-b e^{-\frac{ \rho bx}{a}} (bx)^{-\alpha}\right]\\
&=& -\frac{\rho(1-b)(1-b^{1-\alpha})}{a(\alpha-1)} x^{1-\alpha}-\left(b-1-\frac{\rho x(1-b)}{a}\right)\left[(1-b^{1-\alpha}) x^{-\alpha}-\lambda(1-b)\right]\\
&&+\lambda(1-b)^2-\frac{\rho}{a}(1-b)\frac{x^{1-\alpha}(1-b^{1-\alpha})}{1-\alpha}\\
&&-(1-b)\left(1+\frac{\rho
x}{a}\right)\left[x^{-\alpha}\left(1-\frac{\rho
x}{a}\right)-b\left(1-\frac{\rho
bx}{a}\right)(bx)^{-\alpha}\right]+O(x),
\end{eqnarray*}
where $\lim_{x\to 0} O(1)=0$; so
  $$\lim_{x\to 0}\frac{dH}{dx}=\lim_{x\to 0}\frac{\rho(1-b)(1-b^{2-\alpha})}{a} x^{1-\alpha}=+\infty$$
meaning that the continuous function $H(x)$ increases from zero
when $x\approx 0$ and becomes negative for big values of $x$.

Therefore, equation (\ref{e10}) has a single positive solution
$\bar x$.

(b) Item (i) of Condition \ref{c1} is obviously satisfied (${\cal
D}=\emptyset$).

For Item (ii), we consider the following three cases.

($\alpha$) Let $0<x\le \bar x$. The differential equation
(\ref{e6p}) holds for function $W$ on the interval $0<x<\bar x$.
For these values of $x$,
  \begin{equation}\label{e14}
\mbox{ for any } i\ge 1,~~~W(b^ix)<W(x).
\end{equation}
To prove this, note that function $W(bx)=\tilde W(bx)$ is
increasing (Fig,\ref{figure1}), so that $W(bx)>W(b^2x)>\ldots$ .
Part III of the function $W(x)$ is convex and function $\tilde
W(bx)$ touching smoothly $W(x)$ at point $\bar x$, is concave, so
that $W(bx)=\tilde W(bx)<W(x)$ here. The same inequality holds for
smaller values of $x$ where $W(x)$ decreases (part II) and
$W(bx)=\tilde W(bx)$ increases. Part I of the function $W(x)$ is
obviously bigger than $W(bx)$, too. Thus the Bellman equation
(\ref{e6}) is satisfied on the interval $0<x<\bar x$ and also on
the interval $(0,\bar x]$.

($\beta$) Consider $x\in(\bar x,\bar x/b]$ and denote $x_1$ and
$x_2$ the points of the analytical maximum and minimum of the
function $W(x)=\tilde W(bx)$. (See Fig.\ref{figure1}.)

For $x\in(\bar x,x_1)$ the function $W(x)$ is concave; hence
  $$a\left.\frac{dW}{dx}\right|_x<a\left.\frac{dW}{dx}\right|_{\bar x}=a\left.\frac{d\tilde W}{dx}\right|_{\bar x}=\rho\tilde W(\bar x)-\frac{\bar x^{1-\alpha}}{1-\alpha}+\lambda\bar x=\rho[\tilde W(\bar x)-z(\bar x)].$$
(See formula (\ref{e9p}).) Since $W(x)$ increases starting from
$W(\bar x)=\tilde W(\bar x)$ and $z(x)$ decreases, we have
  $$\left.\frac{dW}{dx}\right|_{x}<\rho[W(x)-z(x)]=\rho W(x)-\left(\frac{x^{1-\alpha}}{1-\alpha}-\lambda x\right),$$
and the Bellman equation (\ref{e6}) is satisfied because here
$W(x)=\tilde W(bx)=W(bx)$ and $W(b^{i+1} x)<W(bx)$ for all $i\ge
1$ because of (\ref{e14}).

For $x\in[x_1,x_2]$ we have $W(x)>z(x)$ and $a\frac{dW}{dx}\le 0$:
remember, $W(x)=\tilde W(bx)$ and the latter function is of type
II for $x\in[x_1,x_2]$. Therefore, again
  $$a\frac{dW}{dx}-\rho\left[W(x)-\frac{1}{\rho}\left(\frac{x^{1-\alpha}}{1-\alpha}-\lambda x\right)\right]<0$$
and the Bellman equation (\ref{e6}) is satisfied.

For $x\in(x_2,\bar x/b]$, we have
  $$a\frac{dW}{dx}=b\left.\frac{d\tilde W}{dx}\right|_{bx}<\left.\frac{d\tilde W}{dx}\right|_{bx}$$
because function $\tilde W$ increases here and $b\in(0,1)$. Next,
  $$\rho[W(x)-z(x)]=\rho[\tilde W(bx)-z(x)]>\rho[\tilde W(bx)-z(bx)]$$
because the function $z(x)$ decreases. Therefore,
  $$a\frac{dW}{dx}-\rho[W(x)-z(x)]<\left.\frac{d\tilde W}{dx}\right|_{bx}-\rho[\tilde W(bx)-z(bx)]=0:$$
$bx\le\bar x$, and, for these values, equation (\ref{e6p}) holds.
We see that the Bellman equation (\ref{e6}) is satisfied.

($\gamma$) Suppose
  $$a\frac{dW}{dx}-\rho[W(x)-z(x)]<0 \mbox{ for } x\in\left(\frac{\bar x}{b^{i-1}},\frac{\bar x}{b^i}\right],$$
for some natural $i\ge 1$. Then, for $x\in\left(\frac{\bar
x}{b^i},\frac{\bar x}{b^{i+1}}\right]$, we have
  $$a\frac{dW}{dx}-\rho[W(x)-z(x)]=ba\left.\frac{dW}{dx}\right|_{bx}-\rho[W(bx)-z(x)].$$
If $\left.\frac{dW}{dx}\right|_{bx}<0$ then the last expression is
negative. Otherwise,
  $$ba\left.\frac{dW}{dx}\right|_{bx}\le a\left.\frac{dW}{dx}\right|_{bx} \mbox{ and } z(x)<z(bx),$$
 so that
  $$a\frac{dW}{dx}-\rho[W(x)-z(x)]<a\left.\frac{dW}{dx}\right|_{bx}-\rho[W(bx)-z(bx)]<0$$
by the induction supposition.

The Bellman equation (\ref{e6}) is satisfied for all $x>0$.

Item (iii) of condition \ref{c1} is also obviously satisfied:
  $${\cal L}^*=[\bar x,\infty);~~v^{f,{\cal L}^*}(x)=v_i \mbox{ if } x\in\left[\frac{\bar x}{b^{i-1}},\frac{\bar x}{b^i}\right).$$

(c) Note that item (iv) of Condition \ref{c1} is not satisfied.
Indeed, there is an admissible control such that, on any time
interval $(T-1,T]$, $x^\pi(t)$ is so close to zero that $e^{-\rho
T}W(x^\pi(t))<-1$. (Remember that $\lim_{x\to 0} W(x)=-\infty$.)

Let us fix an arbitrary $x_0>0$ and modify the reward rate:
  $$\hat c(x)=\left\{\begin{array}{ll}
c(x), & \mbox{ if } x\ge\min\{x_0,b\bar x\}:=\hat x; \\
c(\hat x), & \mbox{ if } x<\hat x.
  \end{array}\right. $$
Note that $\hat c\ge c$. The function $\tilde W(x)$ given by
(\ref{e7}) will change only for $x<\hat x\le x_0$ and remains
increasing in its part I, meaning that this modified function
$\hat W$ satisfies all items (i)--(iii) of Condition \ref{c1}: the
proof is identical to the one presented above. But now Condition
\ref{c1} (iv) is also satisfied because the function $\hat W$ is
bounded. Therefore, according to Theorem \ref{t1}, $\sup_\pi \hat
J(x_0,\pi)=\hat W(x_0)=\hat J(x_0,\pi^*)$, where $\hat J$
corresponds to the reward rate $\hat c$. But
  $$\sup_\pi {J}_\rho(x_0,\pi)\le \sup_\pi \hat J(x_0,\pi)=\hat W(x_0)=W(x_0),$$ and for the feedback policy $\pi^*$, which is independent of $x_0$, we have
  $$W(x_0)=\hat W(x_0)=\hat J(x_0,\pi^*)={J}_\rho(x_0,\pi^*).$$
The last equality holds because, under the feedback policy
$\pi^*$, starting from $x_0$, the trajectory $x^{\pi^*}(t)$
satisfies $x^{\pi^*}(t)\ge \hat x$ for all $t\ge 0$, and in this
region $\hat c=c$.
\end{proof}

\begin{remark}
The above two theorems assert that if the sending rate is smaller
than $\bar{x}$, then do not send any congestion notification, while if the sending
rate is greater or equal to $\bar{x},$ then send (multiple, if
needed) congestion notifications until the sending rate is reduced to some level below
$\overline{x}$ with $\overline{x}$ given by (\ref{eq:threshold})
under the average criterion and by Theorem \ref{t2}(a) under the
discounted criterion. This defines our proposed threshold-based AQM scheme.
\end{remark}

\section{Conclusion}\label{con}
To sum up, in this paper, we studied optimal impulsive control problems
on infinite time interval with both discounted and time average criteria.
We have established Bellman equations and provided conditions for the
verification of canonical triplet. Our general results are then applied
to construct a novel AQM scheme, which takes into account not only the
traffic transiting through the bottleneck links but also the congestion
control algorithms operating at the edges of the network. We are currently
working on practical aspects of the proposed scheme and its validation.
Preliminary results indicate that the new scheme improves fairness significantly
with respect to alternative solutions like the RED algorithm.

\section*{Acknowledgement}
This work is partially funded by INRIA Alcatel-Lucent Joint Lab,
ADR ''Semantic Networking''.
\bibliographystyle{ieeetr}
\bibliography{tcp}

\end{document}